\documentclass[conference,a4paper]{IEEEtran}
\usepackage{amssymb}
\usepackage{amsmath,amsthm,amssymb}
\usepackage{graphicx}
\usepackage{color}
\usepackage{algorithm}  
\usepackage{enumerate}
\usepackage[noend]{algorithmic}
\allowdisplaybreaks
\begin{document}
\renewcommand{\IEEEbibitemsep}{0pt plus 2pt}
\makeatletter
\IEEEtriggercmd{\reset@font\normalfont\footnotesize}
\makeatother
\IEEEtriggeratref{1}
\raggedbottom
\sloppy
\newtheorem{claim}{Claim}
\newtheorem{corollary}{Corollary}
\newtheorem{definition}{Definition}
\newtheorem{example}{Example}
\newtheorem{exercise}{Exercise}
\newtheorem{fact}{Fact}
\newtheorem{lemma}{Lemma}
\newtheorem{note}{Note}
\newtheorem{obs}{Observation}
\newtheorem{problem}{Problem}
\newtheorem{property}{Property}
\newtheorem{proposition}{Proposition}
\newtheorem{question}{Question}
\newtheorem{ru}{Rule}
\newtheorem{solution}{Solution}
\newtheorem{theorem}{Theorem}
\newenvironment{remark}[1]{\textbf{Remark: }}

\newcommand{\A}{{\bf a}}
\newcommand{\CC}{{\cal C}}
\newcommand{\B}{{\bf b}}
\newcommand{\I}{{\bf i}}
\newcommand{\p}{{\bf p}}
\newcommand{\bh}{{\bf h}}
\newcommand{\q}{{\bf q}}
\newcommand{\x}{{\bf x}}
\newcommand{\y}{{\bf y}}
\newcommand{\z}{{\bf z}}
\newcommand{\M}{{\bf M}}
\newcommand{\N}{{\bf N}}
\newcommand{\X}{{\bf X}}
\def\r {{\bf r}}
\newcommand{\bP}{{\bf P}}
\newcommand{\Q}{{\bf Q}}
\newcommand{\U}{{\bf u}}
\def\d {{\tt d}}
\def\D {{\tt D}}
\def\V {{\tt V}}
\def\w{{\bf w}}
\def\Av {{\tt Avg}}
\def\restr{\downharpoonright}
\newcommand{\comment}[1]{\ $[\![${\normalsize #1}$]\!]$ \ }

\newcommand{\restrict}{{\mathbin{\vert\mkern-0.5mu\grave{}}}}

\newcommand{\glb}{{\tt glb}}
\newcommand{\lub}{{\tt lub}}

\newcommand{\XX}{{\cal X}}
\newcommand{\YY}{{\cal Y}}
\newcommand{\LL}{{\cal L}}
\newcommand{\PP}{{\cal P}}
\newcommand{\pv}{{{\bf p}=(p_1,\ldots , p_n)}}

\newcommand{\qv}{{{\bf q}=(q_1,\ldots , q_n)}}
\newcommand{\Inf}{{\bf t}}
\newcommand{\Sup}{{\bf s}}
\newcommand{\inff}{{\tt t}}
\newcommand{\supp}{{\tt s}}
\def\H {{\cal H}}
\def\l {{\bf l}}

\allowdisplaybreaks
\newcommand{\g}{{\bf g}}
\newcommand{\s}{{sup}}
\def\i {{inf}}
\def\r {{\bf r}}
\def\t {{\bf t}}
\newcommand{\pcoin}{{coin of bias $p$ }}
\newcommand{\ptree}{{tree of bias $p$ }}
\newcommand{\remove}[1]{}

\newcommand{\todo}[1]{{\tiny\color{red}$\tiny\circ$}%
{\marginpar{\flushleft\\scriptsizes\sf\color{red}$\bullet$#1}}}

\thispagestyle{plain}

\title{Maximum Entropy Interval Aggregations}
\author{
  \IEEEauthorblockN{Ferdinando Cicalese}
  \IEEEauthorblockA{
    Universit\`a di Verona,
    Verona, Italy\\
    Email: cclfdn22@univr.it } 
  \and
  \IEEEauthorblockN{Ugo Vaccaro}
  \IEEEauthorblockA{
    Universit\`a di Salerno,
    Salerno, Italy\\
    Email: uvaccaro@unisa.it} 
}
	
\maketitle
\begin{abstract}
Given a probability distribution $\p = (p_1, \dots, p_n)$ and an integer $1\leq m < n$, we say that 
$\q = (q_1, \dots, q_m)$ is a {\em contiguous $m$-aggregation of $\p$} if 
 there exist indices $0=i_0 < i_1 < \cdots < i_{m-1} < i_m = n$ such that 
for each $j = 1, \dots, m$ it holds that $q_j = \sum_{k=i_{j-1}+1}^{i_j} p_k.$
In this paper, we consider the problem of efficiently
finding the contiguous $m$-aggregation \emph{of maximum entropy}.
We design a dynamic programming algorithm that solves the problem exactly,
and two more time-efficient greedy algorithms that provide slightly sub-optimal
solutions. We also discuss a few scenarios where our problem matters.

\end{abstract}
\section{Introduction}
The problem of aggregating data in  a compact and meaningful way,
and such  that the aggregated data   retain the maximum possible information  
contained in  
the original data, arises in many scenarios \cite{Gan}. In this 
paper we consider the following particular instance of the general problem.
Let $\XX=\{x_1,  \ldots ,x_n\}$
be a  finite alphabet, and $X$ be any    random variable (r.v.)
taking values in $\XX$ according to the probability distribution
$\p=(p_1, p_2, \ldots , p_n)$, that is, such that $P\{X=x_i\}=p_i>0$, 
for $i=1, 2, \ldots , n$. Consider a partition $\Pi=(\Pi_1, \ldots , \Pi_m)$,
$m<n$, of the alphabet $\XX$, where each class $\Pi_i$ of the partition
$\Pi$ consists of \emph{consecutive} elements of $\XX$. That is,
there exist indices $1\leq  i_1 < \cdots < i_{m-1} < i_m = n$ such that 
$\Pi_1=\{x_1, \ldots , x_{i_1}\}, \Pi_2=\{x_{{i_1}+1}, \ldots , x_{i_2}\}, 
\ldots , \Pi_m=\{x_{{i_{m-1}}+1}, \ldots , x_{i_m}\}$.
Any given such a partition $\Pi=(\Pi_1, \ldots , \Pi_m)$ naturally gives
a r.v. $Y=f_\Pi(X)$, where for each $x\in \XX$ it holds that  $f_\Pi(x)=i$ if
and only if $x\in \Pi_i$.
Let $\q = (q_1, \dots, q_m)$ be the probability distribution of
r.v. $Y$. The values of the 
probabilities $q_j$ can obviously be computed as follows:   
for indices 
$0=i_0 < i_1 < \cdots < i_{m-1} < i_m = n$ it holds that $q_j = \sum_{k=i_{j-1}+1}^{i_j} p_k.$
The problem we consider in this paper is to determine the value
\begin{equation}\label{eq:problem1}
\max_\Pi I(X;f_\Pi(X)),
\end{equation}
where $I$ denotes the mutual information and  the maximum is computed over all $m$-class
partitions $\Pi=(\Pi_1, \ldots , \Pi_m)$ of set $\XX$,
in which each 
class $\Pi_i$ of the partition 
$\Pi$ consists of {consecutive} elements of $\XX$.
Since the function $f_\Pi$ is deterministic,  the problem (\ref{eq:problem1})
can be equivalently stated as
\begin{equation}\label{eq:problem2}
\max_\Pi H(f_\Pi(X)),
\end{equation}
where $H$ denotes Shannon entropy and the maximization takes place over the same domain as in (\ref{eq:problem1}).
The formulation (\ref{eq:problem1}) is common in the area of clustering
(e.g.,  \cite{F+,KMN}) to emphasize that the objective is
to reduce the ``dimension'' of the data (i.e., the cardinality of $|\XX|$) under the constraint that
the ``reduced'' data gives the maximum possible information
towards the original, not aggregated data. We 
remark that, in general, there is no loss of generality in 
considering the problem (\ref{eq:problem1}) for deterministic
functions only (e.g., see 
\cite{GA,kur}).

The  contributions of this paper consist in  efficient algorithms
to solve the optimization problems  (\ref{eq:problem1}) and (\ref{eq:problem2}).
More precisely, we design a dynamic programming  algorithm that runs in time
$O(n^2m)$ to find a partition $\Pi$ that achieves the maximum in 
(\ref{eq:problem2}). Since the time complexity 
$O(n^2m)$ can be too large in some applications, we also provide much more efficient
 greedy algorithms that  
return a solution \emph{provably} very close to the optimal one. 
We remark that the optimization problem (\ref{eq:problem2})
is  strongly NP-hard in case the function $f$ is 
an \emph{arbitrary} function such that $|f(\XX)|=m$, i.e., the partition into $m$ 
classes of $\XX$
induced by $f$ \emph{is not constrained } to contain only   classes made
by contiguous elements of $\XX$ (see \cite{CGV18}).

The rest of the paper is organized as follows. In Section \ref{sec:related}
we discuss the relevance of  our results in the context of related works.
In Section \ref{sec:dynamic} we present our $O(n^2m)$ dynamic programming
algorithm to solve problems (\ref{eq:problem1}) and (\ref{eq:problem2}).
In the final Section \ref{sec:greedy} we present two sub-optimal, but more 
time efficient,  
greedy algorithms for the same problems. 

\section{Related work}\label{sec:related}
The  problem of aggregating data  (or source symbols,
if we think of information sources) in an informative way has been widely studied 
in many different scenarios. One of the motivations is that data aggregation 
is often  an useful, preliminary step to reduce 
the complexity of successive data manipulation.
In this section  we limit ourselves to  point out the work that is strictly  related to ours.

In the paper \cite{KLMM15} the authors considered the following problem.
Given a discrete memoryless source, emitting symbols from
the alphabet $\XX=\{x_1,  \ldots ,x_n\}$ according to  the probability distribution
$\p=(p_1, p_2, \ldots , p_n)$, the question is to find a 
partition $\Pi=(\Pi_1, \ldots , \Pi_m)$,
$m<n$, of  the source alphabet $\XX$ where, as before,  each  $\Pi_i$ 
 consists of {consecutive} elements of $\XX$, and such that
the sum
\begin{equation}\label{eq:problem3}
\frac{1}{m}\sum_{i=1}^m\sum_{j=1}^m |q_i-q_j|,
\end{equation}
is \emph{minimized}. Each 
 $q_j$ in (\ref{eq:problem3}) is the sum of the probabilities $p_k$'s 
corresponding to the elements  $x_k\in \XX$ that belong to $\Pi_j$,
that is our $q_j=\sum_{k=i_{j-1}+1}^{i_j} p_k.$
The motivation of the authors of \cite{KLMM15} to study above problem 
is that the 
minimization of expression  (\ref{eq:problem3}) constitutes  the \emph{basic step}
in  the well known Fano algorithm \cite{Fano} for $m$-ary variable length encoding 
 finite-alphabet memoryless source.
In fact, solving (\ref{eq:problem3}) allows one to 
find a partition of $\XX$ such that the cumulative
probabilities of each class partition are \emph{as similar as possible}.
Obviously, the basic step has to be iterated in each class $\Pi_i$, till  the partition
is made by singletons.
Now, it is not hard to see that
\begin{equation}\label{eq:problem4}
\frac{1}{m}\sum_{i=1}^m\sum_{j=1}^m |q_i-q_j|=2+\frac{2}{m}-\frac{4}{m}\sum_{i=1}^m iq_{[i]},
\end{equation}
where $(q_{[1]}, \ldots , q_{[m]})$ is the vector that contains
the same elements as $\q=(q_1, \ldots, q_m)$, but now ordered
in non-increasing fashion. From equality (\ref{eq:problem4}) one can see
that the problem of minimizing
expression  (\ref{eq:problem3}), over all partitions as stated above,
is \emph{equivalent to maximizing} the quantity $\sum_{i=1}^m iq_{[i]}$ over
the same domain. The quantity $\sum_{i=1}^m iq_{[i]}$ 
is the well known guessing entropy by J. Massey \cite{Massey}.
Therefore, while in our problem (\ref{eq:problem2}) we seek a 
partition of $\XX$ such that the cumulative
probabilities of each class partition are {as similar as possible},
and the measure we use to appraise this quality 
is the Shannon entropy, the authors of \cite{KLMM15} address the
same problem 
using   the  guessing entropy, instead (this observation  is not present in
\cite{KLMM15}). We should add that the criterion (\ref{eq:problem3})
 used in \cite{KLMM15}
allows the authors to prove that the Fano algorithm produces an  
 $m$-ary variable length encoding 
of the given source such that the average length of the encoding
is strictly smaller than $\frac{H(\p)}{\log m} +1 - p_{\min}$,
for $m=2$ and $m=3$ (and they conjecture that this is true also for
any $m\geq 4$), where $\p$ is the source probability distribution and
$p_{\min}$ is the probability of the least likely source symbol.
On the other hand, \emph{it is not} clear how to efficiently solve
the optimization problem (\ref{eq:problem3}). In fact, it is not known whether it 
 enjoys or not  the \emph{optimal substructure} property, 
a necessary condition so that the problem could be optimally solved with 
known techniques like dynamic programming, greedy, etc. \cite{Cormen}.
As mentioned before, our problem (\ref{eq:problem2}) \emph{can 
be} optimally  solved  via dynamic programming. 
Numerical simulation suggests that  optimal solutions
to our problem (\ref{eq:problem2}) can be used to construct Fano encodings with
the same upper bound on the average length as the ones
constructed in \cite{KLMM15}.

A similar  question,  in which the aggregation operations of the elements of $\XX$   are 
again constrained 
by  given  rules, was considered in 
\cite{CK}.
There, the authors  consider
the problem of constructing the  \emph{summary tree}  of a given weighted tree,
by means of  contraction operations on trees.
{Two types of contractions are allowed: 1) subtrees may be contracted to single nodes that represent
the corresponding subtrees, 2) 
 subtrees whose roots
are siblings may be contracted to single nodes.  
Nodes obtained by contracting subtrees have weight equal to the sum
of the node weights in the original contracted subtrees.
} Given a bound on the number of nodes in the 
resulting summary tree, the problem 
is to compute the summary tree of \emph{maximum entropy},
where the entropy of a tree is the Shannon entropy of the normalized node weights. 
In \cite{ME} the authors consider the problem of quantizing a finite
alphabet $\XX$ by collapsing   properly chosen contiguous sequences of symbols of $\XX$ (called
\emph{convex codecells} in \cite{ME}) to  single elements. The objective   is to  minimize the expected 
distortion induced by the quantizer, for some  classes of distortion measures. 
Our similar scenario  would correspond to the minimization of $H(X)-H(f_\Pi(X))$,
not considered in \cite{ME}.

Our results could find applications also in data compression  for sources with large alphabet (e.g. \cite{Moffatt}).
One could use our techniques as a pre-processing phase to reduce the source alphabet from
a large one to a smaller one, in order to obtain a new source that retains most of the 
entropy  as the original one, just because of (\ref{eq:problem2}). An encoding
of the so constructed ``reduced source'' can be easily 
transformed to an encoding of the original source by
exploiting the fact that the partition of the original 
source alphabet has been performed with consecutive 
subsets of symbols.
Finally, other problems similar to ours were considered in papers 
\cite{KK99,La+}.
It seems that our findings could be useful in ``histogram compression'',
where the constraint that one can merge only adjacent class intervals is natural \cite{PM86}.

 \section{An optimal dynamic programming algorithm}\label{sec:dynamic}
We find it convenient to formulate  problems (\ref{eq:problem1}) and (\ref{eq:problem2})
in a slightly different language. We  give the following definition.

\begin{definition}
Given a $n$-dimensional vector of strictly positive numbers $\p = (p_1, \dots, p_n)$ and a positive integer $m < n$, we say that 
a vector $\q = (q_1, \dots, q_m)$ is a {\em contiguous $m$-aggregation of $\p$} if the following condition hold:
 there exist indices $0=i_0 < i_1 < \cdots < i_{m-1} < i_m = n$ such that for each $j = 1, \dots, m$ it holds that $q_j = \sum_{k=i_{j-1}+1}^{i_j} p_k.$
\end{definition}

Thus, our problems can be so formulated:

\medskip
\noindent
{\bf Problem Definition.} Given an $n$-dimensional probability distribution $\p = (p_1, \dots, p_n)$ (where all components are assumed to be 
strictly positive)
and an integer $1\leq m < n$, find a contiguous $m$-aggregation of $\p$ of \emph{maximum entropy}.

\medskip\noindent
Our dynamic programming algorithm proceeds as follows.
For $j = 1, \dots, n,$ let $s_j = \sum_{k=1}^j p_k.$ Notice that we can compute all these values in $O(n)$ time. 
For a sequence of numbers $\w = w_1, \dots, w_t$ such that for each $i=1,\dots, t, \, w_i \in (0,1]$ and $\sum_{i=1}^t w_i \leq 1,$ 
we define the {\em entropy-like sum of $\w$} as $\tilde{H}(\w) = -\sum_{j=1}^t w_t \log w_t.$ 
Clearly when $\w$ is a probability distribution we have that the entropy-like sum of $\w$ coincides with the Shannon entropy of $\w.$
For each $i=1, \dots, m$ and $j = 1, \dots, n$ let $hq[i,j]$ be the maximum entropy-like sum of a contiguous  $i$-aggregation of the sequence $p_1, \dots, p_j.$
Therefore, $hq[m,n]$ is the sought  maximum entropy of a contiguous $m$-aggregation of $\p.$
Let $\hat{\q} = (q_1, \dots, q_i)$ be a contiguous $i$-aggregation of $(p_1, \dots, p_j)$ of maximum entropy-like sum. 
Let $r$ be the index such that $q_i = \sum_{k=r}^j p_k.$
We have $q_i = s_j - s_{r-1}$ and 

\smallskip
$~~~~~~~~~~\displaystyle \tilde{H}(\hat{\q}) = -(s_j-s_{r-1}) \log (s_j-s_{r-1}) + \tilde{H}(\q'),$ 

\smallskip\noindent
where $\q' = (q_1, \dots, q_{i-1})$. 
Now we observe that $\q'$ is a
contiguous $(i-1)$-aggregation of $(p_1, \dots, p_{r-1}).$ Moreover, since $\tilde{H}(\hat{\q})$ is \emph{maximum}---among the entropy-like sum of any
contiguous $i$-aggregation of $(p_1, \dots, p_i)$ ---it must also hold that $\tilde{H}(\q')$ is maximum among any contiguous $(i-1)$-aggregation of 
$(p_1, \dots, p_{r-1}).$ Based on this observation we can compute the $hq[\cdot,\cdot]$ values recursively as follows:

$$
hq[i,j] = \begin{cases} 
\displaystyle{\max_{k=i, \dots, j}} \{hq[i-1, k-1] & \\
~~~- (s_j - s_{k-1}) \log (s_j - s_{k-1})\} & i >1, \, j \geq i\\
-s_j \log s_j & i = 1.
\end{cases}
$$

There are $n \times m$ values to be computed and each one of them can be computed in $O(n)$ (due to the $\max$ in the first case). Therefore the
computation of $h[m,n]$ requires $O(n^2 m)$ time. 
By a standard procedure, once one has  the whole table $hq[\cdot,\cdot]$,  one  can reconstruct the contiguous $m$-aggregation of $\p$ achieving
entropy $hq[m,n]$ by backtracking on the table.

\section{Sub-optimal greedy algorithms}\label{sec:greedy}

We start by recalling  a few  notions of majorization theory \cite{MO}
that  are relevant to our context.

\begin{definition}\label{defmaj} 
Given two probability distributions
$\A=(a_1, \ldots ,a_n)$ and $\B=(b_1, \ldots , b_n)$ with $a_1\geq \ldots \geq a_n\geq 0$ and 
$b_1\geq \ldots \geq b_n\geq 0$, $\sum_{i=1}^na_i=\sum_{i=1}^nb_i=1$, we say that $\A$ is 
{\em majorized} by $\B$, and write  $\A \preceq \B$,
if and only if
$\sum_{k=1}^i a_k\leq \sum_{k=1}^i b_k, \quad\mbox{\rm for all }\  i=1,\ldots , n.$
\end{definition}
We   use the majorization 
relationship between vectors of unequal lengths, by properly padding the shorter
one with the appropriate number of $0$'s at the end.
Majorization  induces  a lattice structure  on 
$\PP_n=\{(p_1,\ldots, p_n)\ :  \sum_{i=1}^n p_i=1, \ p_1\geq \ldots \geq p_n\geq 0\}$, see \cite{CV}.%
\remove{  of all ordered probability vectors of $n$ elements, that is, for each $\x,\y,\z\in\PP_n$
it holds that 
\textbf{1)} $\x\preceq \x$;
\textbf{2)} $\x\preceq \y$ and $\y\preceq \z$  implies $\x\preceq \z$;
\textbf{3)}  $\x\preceq \y$ and $\y\preceq \x$  implies $\x =\y$.
}
\remove{It turns out that   that the partially ordered set $(\PP_n,\preceq)$ is indeed a \emph{lattice} \cite{CV},\footnote{The same result was independently 
rediscovered in \cite{cuff}}
i.e., for all $\x,\y  \in \PP_n$
there exists
a unique
{\em least upper bound}  $\x\lor \y$ and 
 a unique {\em greatest lower bound}  $\x \land \y$.
}
Shannon entropy function enjoys the  important Schur-concavity property \cite{MO}:
\emph{For any  $\x,\y\in\PP_n$, 
 $\x\preceq \y$ implies that    $H(\x)\geq H(\y)$.}
		\remove{
A notable strengthening  of above fact has been proved  in \cite{HV}. There, the authors prove that 
$\x\preceq \y$ implies $H(\x)\geq H(\y)+D(\y||\x)$, where $D(\y||\x)$ is the relative entropy
between $\x$ and $\y$.}
We also need the  concept of \emph{aggregation}  and a result from \cite{CGV17}.
Given $\p=(p_1, \ldots , p_n)\in \PP_n $ and an integer $1\leq m<n$,
we say that $\q=(q_1, \ldots , q_m)\in \PP_m$ is an \emph{aggregation} 
of $\p$ if there is a partition of $\{1, \ldots , n\}$ into disjoint sets $I_1, \ldots , I_m$
such that $q_j=\sum_{i\in I_j}p_i$, for $j=1, \ldots m$.

\begin{lemma}{\rm\cite{CGV17}}\label{pprecq}
Let $\q\in \PP_m $ be \emph{any} aggregation of $\p\in \PP_n$.  Then it holds that
$\p\preceq \q$.
\end{lemma}
\remove{\begin{IEEEproof} 
We shall prove  by induction on $i$ that $\sum_{k=1}^{i}q_k\geq \sum_{k=1}^{i}p_k$.
Because  $\q$ is an aggregation of $\p$, we know that there exists    $I_j\subseteq \{1, \ldots , n\}$
such that $1\in I_j$. This implies that  $q_1\geq q_j\geq p_1$. Let us suppose that 
$\sum_{k=1}^{i-1}q_k\geq \sum_{k=1}^{i-1}p_k$. 
If there exist  indices $j\geq i$ and $\ell\leq i$ such that $\ell\in I_j$,
then $q_i\geq q_j\geq p_\ell\geq p_i$,  implying 
$\sum_{k=1}^{i}q_k\geq \sum_{k=1}^{i}p_k$.  
Should it be otherwise,  for each $j\geq i$ and $\ell\leq i$ it holds that $\ell\not \in I_j$. Therefore,
$\{1, \ldots ,i\}\subseteq I_1\cup \ldots \cup I_{i-1}$. This immediately gives
$\sum_{k=1}^{i-1}q_k\geq \sum_{k=1}^{i}p_k$, from which we 
get $\sum_{k=1}^{i}q_k\geq \sum_{k=1}^{i}p_k$.
\end{IEEEproof}

\medskip
}

\remove{
\section{Aggregations with partitions made of index intervals}

From now on, unless explicitly said, we do not assume that the components of the probability distributions are 
listed in non-increasing order, as was the case in the previous part of this manuscript.

\begin{definition}
Given a $n$-dimensional vector of positive numbers $\p = (p_1, \dots, p_n)$ and a positive integer $m < n$, we say that 
a vector $\q = q_1, \dots, q_m$ is a {\em contiguous $m$-aggregation of $\p$} if the following conditions hold:
\begin{itemize}
\item $q_i > 0$ for each $i=1, \dots, m$
\item there exist indices $0=i_0 < i_1 < \cdots < i_{m-1} < i_m = n$ such that for each $j = 1, \dots, m$ it holds that $q_j = \sum_{k=i_{j-1}+1}^{i_j} p_k.$
\end{itemize}
\end{definition}

We study the following problem:

\medskip

\noindent
{\bf Problem Definition.} Given an $n$-dimensional probability distribution $\p = (p_1, \dots, p_n)$ (where all components are assumed to be positive)
and a positive integer $m < n$, find a contiguous $m$-aggregation of $\p$ of maximum entropy.

\bigskip

We will first show an optimal DP algorithm for the problem which takes $O(n^2 m )$ time. Then we will analyze a linear, i.e., $O(n)$-time greedy algorithm which 
guarantees to find a contiguous $m$-aggregation $\q$ of $\p$ such that $H(\q) \geq H(\q^*) - \frac{2}{e \ln(2)}$ where $\q^*$ is a contiguous $m$-aggregation 
of $\p$ of maximum entropy. In words, the greedy algorithm guarantees an additive approximation which is not worse than $\frac{2}{e \ln(2)} \sim 1.0614757.$ 

\subsection{The Optimal DP algorithm}

For $j = 1, \dots, n,$ let $s_j = \sum_{k=1}^j p_k.$ Notice that we can compute all these values in $O(n)$ time. 

For a sequence of numbers $\w = w_1, \dots, w_t$ such that for each $i=1,\dots, t, \, w_i \in (0,1]$ and $\sum_{i=1}^t w_i \leq 1,$ 
we define the {\em entropy-like sum of $\w$} as $\tilde{H}(\w) = -\sum_{j=1}^t w_t \log w_t.$ 
Clearly when $\w$ is a probability distribution we have that the entropy-like sum of $\w$ coincides with the Shannon entropy of $\w.$

For each $i=1, \dots, m$ and $j = 1, \dots, n$ let $hq[i,j]$ be the maximum entropy-like sum of a contiguous  $i$-aggregation of the sequence $p_1, \dots, p_j.$

Therefore, $hq[m,n]$ is the maximum entropy of a contiguous $m$-aggregation of $\p.$

Let $\hat{\q} = q_1, \dots, q_i$ be a contiguous $i$-aggregation of $p_1, \dots, p_j$ of maximum entropy like sum. 
Let $r$ be the index such that $q_i = \sum_{k=r+1}^j p_k.$
We have $q_i = s_j - s_{r-1}$ and $\tilde{H}(\hat{\q}) = -(s_j-s_{r-1}) \log (s_j-s_{r-1}) + \tilde{H}(\q')$ where $\q' = q_1, \dots, q_{i-1}$. 
Now we observe that $\q'$ is 
contiguous $(i-1)$-aggregation of $p_1, \dots, p_{r-1}.$ Moreover, since $\tilde{H}(\hat{\q})$ is maximum---among the entropy-like sum of any
contiguous $i$-aggregation of $p_1, \dots, p_i$, it must also hold that $\tilde{H}(\q')$ is maximum among any contiguous $(i-1)$-aggregation of 
$p_1, \dots, p_{r-1}.$

Based on this observation we can compute the $hq[\cdot,\cdot]$ values recursively as follows:
$$
hq[i,j] = \begin{cases} 
\displaystyle{\max_{k=i, \dots, j}} \{hq[i-1, k-1] & \\
~~~~~~- (s_j - s_{k-1}) \log (s_j - s_{k-1})\} & i >1, \, j \geq i\\
-s_j \log s_j & i = 1
\end{cases}
$$

There are $n \times m$ values to be computed and each one of them can be computed in $O(n)$ (due to the $\max$ in the first case). Therefore the
computation of $h[m,n]$ requires $O(n^2 m)$ time. 
By a standard procedure, once we have computed the whole table of $hq[]$ values we can reconstruct the contiguous $m$-aggregation of $\p$ achieving
entropy $hq[m,n]$ by backtracking on the table.
 
}
We now  present our first greedy   approximation algorithm for the problem of 
finding the maximum entropy contiguous $m$-aggregation of
a given probability distribution $\p=(p_1, \dots, p_n).$
The pseudocode of the algorithm is given
below.

\begin{algorithm}[ht]
\small
{\sc  Greedy-Approximation}$(p_1, \dots p_n,m)$~~
\begin{algorithmic}[1]
\STATE{//Assume $n > m$ and an auxiliary value $p_{n+1} = 3/m$}
\STATE{$i \leftarrow 0,\, j \leftarrow 1$}    \label{phase1-start}
\STATE{$partialsum \leftarrow p_j$}
\WHILE{$j \leq  n$}
\STATE{$i \leftarrow i+1, \, start[i] \leftarrow j$}
\WHILE{$partialsum + p_{j+1} \leq 2/m$}
\STATE{$partialsum \leftarrow partialsum + p_{j+1}, \, j \leftarrow j+1$}
\ENDWHILE
\STATE{$q_i \leftarrow partialsum, \, end[i] \leftarrow j$}
\STATE{$j \leftarrow j+1,\, partialsum \leftarrow p_j$} \label{phase1-end}
\ENDWHILE
\STATE{// At this point $i$ counts the number of components in $\q$}
\STATE{// If $i < m$ we are going to split exactly $m-i$ components}
\STATE{$k \leftarrow m-i, \, j \leftarrow 1$}   \label{phase2-start}
\WHILE{$k > 0$}
\WHILE{$start[j] = end[j]$}
\STATE{$j \leftarrow j+1$}
\ENDWHILE
\STATE{$i \leftarrow i+1, \, k \leftarrow k-1$}
\STATE{$start[i] \leftarrow start[j], \, end[i] \leftarrow start[j],\, start[j] \leftarrow start[j]+1$}  \label{phase2-end}
\ENDWHILE
\end{algorithmic}
\caption{A  linear time greedy approximation algorithm}
\label{algo:Greedy}
\end{algorithm}
 
The algorithm has two phases. 
In the first phase, lines from \ref{phase1-start} to \ref{phase1-end}, 
the algorithm iteratively builds a new component of $\q$ as follows: 
Assume that the first $i$ components of $\q$ have been produced by aggregating the first $j$ components of $\p.$
If $p_{j+1} > 2/m$ then $q_{i+1}$ is the aggregation of the singleton interval containing only $p_{j+1}.$
Otherwise, $q_{i+1}$ is set to be the aggregation of the largest number of components $p_{j+1}, p_{j+2}, \dots$
such that their sum is not larger than $2/m.$   

For each $k = 1, \dots, i,$ the values $start[k]$ and $end[k]$ 
are meant to contain the first and the last component of $\p$ which are 
aggregated into $q_k$. By construction,
we have that $start[k] \neq end[k]$ indicates that $q_k \leq  2/m.$
The first crucial observation is that, at the end of the first phase, the number $i$
 of components in the distribution $\q$ under construction is smaller than $m.$ To see this, 
it is enough to observe that by construction $q_j + q_{j+1} > 2/m,$ for any $j=1,2, \dots, \lfloor i/2 \rfloor.$ 
Therefore,  arguing by contradiction, if we had $i \geq m+1$ we would reach the following counterfactual inequality 
$$1\! =\! \sum_{j=1}^i q_j \!\geq\! \sum_{j=1}^{\lfloor i/2 \rfloor} (q_{2j-1} + q_{2j}) \! >\!\!\sum_{j=1}^{\lfloor i/2 \rfloor} \frac{2}{m} 
\!=\! \left\lfloor \frac{i}{2} \right\rfloor \frac{2}{m} 
\!\geq\! \frac{2i-2}{2m}\geq 1.$$
In the second phase, lines \ref{phase2-start}-\ref{phase2-end}, 
the algorithm splits the first $m-i$ components 
of $\q$ which are obtained by aggregating at least two components of $\p$. Notice that, as observed above, 
such components of $\q$ are not larger than $2/m.$ Hence, also the resulting components in which 
they are split have size at most $2/m.$
It is important to notice that there must exist at least $m-i$ such ``composite''\footnote{We are calling a component $q_j$ 
{\em composite} if it is obtained as the sum of at least two components of $\p.$} 
components, because of the assumption $n > m,$
and the fact that each component of $\p$ is non zero. 
As a result of the above considerations, the aggregation $\q$ returned by the {\sc Greedy-Approximation} algorithm 
can be represented, after reordering its
components in non-increasing order, as
$\q = (q_1, \dots, q_{k^*}, q_{k^*+1}, \dots q_m),$
where $q_1, \dots, q_{k^*}$ are all larger than $2/m$ and coincide with the $k^*$ largest components of $\p$
and the remaining components of $\q$, namely $q_{k^*+1}, \dots, q_m,$ are all not larger than $2/m.$
Let us now define the quantities 
 $A = 1 - \sum_{j=1}^{k^*} q_j, \quad \mbox{and} \quad B =  \sum_{j=1}^{k^*} q_j \log \frac{1}{q_j}.$\\
It holds that 
\begin{eqnarray}
H(\q) &=& \sum_{j=1}^{k^*} q_j \log \frac{1}{q_j} + \sum_{j=k^*+1}^{m} q_j \log \frac{1}{q_j}   \label{q-entropy-0} \\
&=& B +  \sum_{j=k^*+1}^{m} q_j \log \frac{1}{q_j}    \label{q-entropy-1} \\
&\geq& B +  \sum_{j=k^*+1}^{m} q_j \log \frac{m}{2}    \label{q-entropy-2} \\
&=& B +  A \log(m) - A  \label{q-entropy-3} 
\end{eqnarray}
where (\ref{q-entropy-1}) follows by definition of $B$; (\ref{q-entropy-2}) follows by the fact that
$q_j \leq \frac{2}{m}$ for any $j > k^*$; (\ref{q-entropy-3}) follows by definition of $A$ and the 
basic properties of the logarithm.

\begin{lemma} \label{lemma:qVSq-tilde}
Let $\tilde{\q}$ be the probability distribution defined as
$\tilde{\q} = (q_1, \dots, q_{k^*}, \frac{A}{m-k^*}, \dots,  \frac{A}{m-k^*}).$ Then, it holds that: \qquad 
$\displaystyle ~~~~~~~~~~~~~~~~~~~~~~~~ H(\q) \geq H(\tilde{\q}) - \frac{2}{e \ln(2)}.$
\end{lemma}
\begin{proof}
We have 
\begin{eqnarray*}
H(\tilde{\q}) &=& \sum_{j=1}^{k^*} q_j \log\frac{1}{q_j} + \sum_{j=k^*+1}^m \frac{A}{m-k^*} \log \frac{m-k^*}{A} \\
&=& B + A \log(m-k^*) - A \log A.
\end{eqnarray*}
Therefore, by using the above lower bound (\ref{q-entropy-0})-(\ref{q-entropy-3}) on the entropy of $\q$ it follows that 
\begin{eqnarray*}
H(\tilde{\q}) - H(\q) &\leq& A \log\frac{m-k^*}{m} - A \log(A) + A  \\
&\leq& - A \log(A) + A  \leq \frac{2}{e \ln(2)}  
\end{eqnarray*}
where the second inequality follows since $A \log\frac{m-k^*}{m} \leq 0$ for any $k^* \geq 0$ and the last inequality 
follows by the fact that  $A \in [0,1]$ and the maximum of the function $-x\log x + x$ in the interval $[0,1]$ is 
$\frac{2}{e \ln(2)}.$
\end{proof}

Let $\q^* = (q_1^*, \dots q_m^*)$ be a contiguous $m$-aggregation of $\p$ of maximum entropy.
We can use $\tilde{\q}$ to compare the entropy of our greedily constructed contiguous $m$-aggregation $\q$ to the entropy of $\q^*$.
We prepare the following
\begin{lemma} \label{lemma:q-tildeVSq-star}  
It holds that $\tilde{\q} \preceq \q^*,$ therefore $H(\tilde{\q}) \geq H(\q^*).$
\end{lemma}
\begin{proof}
Assume, w.l.o.g., that the components of $\q^*$ are sorted in non-increasing order. 
Let $\tilde{\p}= (\tilde{p}_1, \dots, \tilde{p}_n)$ be the probability distribution obtained by reordering the components of $\p$ in non-increasing order.
It is not hard to see that, by construction, we have $\tilde{p}_j = q_j$ for each $j = 1, \dots, k^*.$
Since $\q^*$ is an aggregation of $\tilde{\p}$, by Lemma \ref{pprecq}, we have that $\tilde{\p} \preceq \q^*$, which immediately implies 
\begin{equation} \label{q-star:first-k}
\sum_{s=1}^j q_s = \sum_{s = 1}^j \tilde{p}_s \leq \sum_{s=1}^j q^*_j \qquad \mbox{for each }j=1, \dots, k^*.
\end{equation}
Moreover, by the last inequality with $j=k^*$ it follows that
$\sum_{s=k^*+1}^m q^*_s \leq 1- \sum_{s=1}^{k^*} q_s = A.$
This, together with the assumption that $q^*_1 \geq \cdots q^*_{k^*} \geq q^*_{k^*+1} \geq \cdots q^*_m$ implies that 
\begin{equation} \label{eq:intermediate}
\sum_{s=t+1}^m q^*_s \leq \frac{m-t}{m-k^*} A \qquad \mbox{for any } t \geq k^*.
\end{equation}
Then, for each $j = k^*, \dots, m$ we have 
$$\sum_{s=1}^j q^*_j = 1 - \!\!\sum_{s=j+1}^m q^*_s \geq 1 -\!\! \frac{m-j}{m-k^*}A = 1 - \!\!\sum_{s=j+1}{m} \tilde{q}_s = 
\sum_{s=1}^j \tilde{q}_s$$
that together with (\ref{q-star:first-k}) implies $\tilde{\q} \preceq \q^*.$

\smallskip
This concludes the proof of the first statement of the Lemma. 
The second statement immediately follows by the Schur concavity of the entropy function.
\end{proof}
We are now ready to summarize our findings.
\begin{theorem}
Let $\q$ be the contiguous $m$-aggregation of $\p$ returned by {\sc  Greedy-Approximation}. Let $\q^*$ be a
contiguous $m$-aggregation of $\p$ of maximum entropy. Then, it holds that
$$H(\q) \geq H(\q^*) - \frac{2}{e \ln(2)} = H(\q^*) - 1.0614756...$$
\end{theorem}
\begin{proof}
Directly from Lemmas \ref{lemma:q-tildeVSq-star} and  \ref{lemma:qVSq-tilde}.   
\end{proof}
\subsection{A slightly improved greedy approach}
We can improve the approximation guarantee of Algorithm \ref{algo:Greedy} by a refined greedy approach of complexity  $O(n + m\log m)$. 
The new idea  is to build the components of $\q$ in such a way that they are either not larger than $3/2m$ or they 
coincide with some large component of $\p$. 
More precisely, when building a new component of $\q$, say $q_i$, the algorithm puts together consecutive components of $\p$ as long as their sum, denoted 
$partialsum$, is not larger than $1/m$. If, when trying to add  the next component, say $p_j$,  the total sum becomes larger than $1/m$ the following three cases are considered:

\noindent
{\em Case 1.} $partialsum + p_j \in [\frac{1}{m}, \frac{3}{2m}].$

In this case $q_i$ is set to include also $p_j$ hence becoming a component of $\q$ of size not larger than $3/2m.$

\noindent
{\em Case 2.} $partialsum + p_j > \frac{2}{m}.$

In this case we produce up to two components of $\q$. Precisely, if $partialsum = 0$ that is $p_j > 2/m$ we set $q_i = p_j$ and only one new component is created. Otherwise,  $q_i$ is set to $partialsum$ (i.e., it is the sum of the interval up to $p_{j-1}$, and it is not larger than $1/m$ and 
$q_{i+1}$ is set to be equal to $p_j.$ Notice that in this case $q_{i+1}$ might be larger than $3/2m$ but it is a non-composite component.

\noindent
{\em Case 3.} $partialsum + p_j \in  (\frac{3}{2m}, \frac{2}{m}).$

In this case we produce one component of $\q$, namely  $q_i$ is set to $partialsum+p_j$ and we mark it.

We first observe that the total number of components of $\q$ created by this procedure is not larger than $m$. 
More precisely, let $k_1, k_2, k_3$ be the number of components created by the application of Case 1, 2, and 3 respectively. 
Each component created by Case 1 has size $\geq 1/m$. When we apply Case 2 we create either one component of size $> 2/m$ or two components of total sum
$> 2/m$. Altogether the $k_2$ components created by Case 2 have total sum at least $k_2/m.$ Then, since each component created by applying Case 3 has 
size at least $3/2m$ we have that
$k_3 \leq \frac{1 - (k_1 + k_2)/m}{3/2m} = \frac{2(m - k_1 - k_2)}{3},$ hence
$k_1 + k_2 + k_3 \leq \frac{2m}{3} + \frac{1}{3}(k_1+k_2)$,
from which we get
\textbf{1)} $m - k_1 - k_2 \geq \frac{3}{2} k_3,$
 and 
\textbf{2)} $m - k_1 - k_2 - k_3 \geq \frac{1}{2} k_3.$
Inequalities \textbf{1)} and \textbf{2)}   mean that if $k_3 > 0$ then the number of components created is smaller than $m$ by a quantity which equals at least half of $k_3.$
In other words, we are allowed to split at least half of the $k_3$ components created by Case 3 and the resulting total number of components
will still be not larger than $m$.
In the second phase of the algorithm, the largest components created from Case 3 are split. 
As a result of the above considerations, the final distribution $\q$ returned by the algorithm has: (i)  components $> 2/m$ which are 
singletons, i.e., coincide with components of $\p$; the remaining components can be divided into two sets, the components of size $> 3/2m$ and the ones of
size $\leq 3/2m$ with the second set having larger total probability mass. In formulas, we can represent the probability vector $\q$, after reordering its
components in non-increasing order, as
$\q = (q_1, \dots, q_{k^*}, q_{k^*+1}, \dots q_{j^*}, q_{j^*+1}, \dots, q_m),$
where: (i) $q_1, \dots, q_{k^*}$ are all larger than $2/m$ and coincide with the $k^*$ largest components of $\p$;
(ii) $q_{k^*+1}, \dots, q_{j^*}$ are all in the interval $(3/2m, 2/m)$; (iii)    $q_{j^*+1}, \dots, q_m,$ are all not larger than $3/2m.$
Let us define the quantities \\
$\displaystyle ~~~~~ A_1 = \sum_{s=k^*1}^{j^*} q_s, \quad  A_2 = \sum_{s=j^*1}^{m} q_s, \quad B =  \sum_{j=1}^{k^*} q_j \log \frac{1}{q_j}.~~~~~$
Let $A = A_1 + A_2.$ Since the algorithm splits the largest components of size $3/2m$ it follows that $A_2 \geq A/2.$
Then, by proceeding like in the previous section we have 
\begin{eqnarray}
\!\!\!\!H(\q)\!\!\!\!\!\! &=&\!\!\!\!\!\!\sum_{s=1}^{k^*} q_s \log \frac{1}{q_s}\! +\!\!\!\sum_{s=k^*+1}^{j^*} \!\!\!q_s \log \frac{1}{q_s} \!  +\!\!\!
\sum_{s=j^*+1}^{m} \!\!\!q_s \log \frac{1}{q_s} \label{qq-entropy-0} \\
&\geq&\!\!\!\! B +  \sum_{s=k^*+1}^{j^*} q_s \log \frac{m}{2}   +  \sum_{s=j^*+1}^{m} q_s \log \frac{2m}{3}  \label{qq-entropy-2} \\
&=& \!\!\!\!B +  (A_1+A_2) \log(m) - A_1 - A_2\log \frac{3}{2} \label{qq-entropy-3} \\
&\geq&\!\!\!\! B + A \log(m) - \frac{A}{2} \log(3)
\end{eqnarray}
 \remove{
\begin{eqnarray}
H(\q) &=& \sum_{s=1}^{k^*} q_s \log \frac{1}{q_s} + \sum_{s=k^*+1}^{j^*} q_s \log \frac{1}{q_s}  \\
& & ~~~~~~~~~~ + \sum_{s=j^*+1}^{m} q_s \log \frac{1}{q_s} \label{qq-entropy-0} \\
&\geq& B +  \sum_{s=k^*+1}^{j^*} q_s \log \frac{m}{2}   +  \sum_{s=j^*+1}^{m} q_s \log \frac{2m}{3}  \label{qq-entropy-2} \\
&=& B +  (A_1+A_2) \log(m) - A_1 - A_2\log \frac{3}{2} \label{qq-entropy-3} \\
&\geq& B + A \log(m) - \frac{A}{2} \log(3)
\end{eqnarray}
}
where the last inequality holds since $A_2 \geq A/2.$
Proceeding like in  Lemma \ref{lemma:qVSq-tilde} above, we have the following result. 

\begin{lemma} \label{lemma:qVSq-tilde-new}
Let $\tilde{\q}$ be the probability distribution defined as
$\tilde{\q} = (q_1, \dots, q_{k^*}, \frac{A}{m-k^*}, \dots,  \frac{A}{m-k^*}).$ It holds that: 
$~~~~~~~~~~~~~~~~~~~~~~\displaystyle H(\q) \geq H(\tilde{\q}) - \frac{\sqrt{3}}{e \ln(2)}.$
\end{lemma}
 
This result, together with Lemma \ref{lemma:q-tildeVSq-star} implies

\begin{theorem}
Let $\q$ be the contiguous $m$-aggregation of $\p$ returned by the algorithm {\sc Greedy-2}. Let $\q^*$ be a
contiguous $m$-aggregation of $\p$ of maximum entropy. Then, it holds that
$$H(\q) \geq H(\q^*) - \frac{\sqrt{3}}{e \ln(2)} = H(\q^*) - 0.91926... .$$
\end{theorem}
\begin{algorithm}[ht]
\small
{\sc  Greedy2}($p_1, \dots, p_n, m$){~// assume $n > m$ and auxiliary $p_{n+1} = 2$}

\begin{algorithmic}[1]
\STATE{$i \leftarrow 0,\, j \leftarrow 1$}    \label{algo2:phase1-start}
\WHILE{$j \leq  n$}
\STATE{$i \leftarrow i+1, \, start[i] \leftarrow j, \, partialsum \leftarrow 0$}
\WHILE{$partialsum + p_{j} \leq 1/m$ {\bf and} $j \leq n$}
\STATE{$partialsum \leftarrow partialsum + p_{j}, \, j \leftarrow j+1$}
\ENDWHILE
\IF{$j > n$}
\STATE{$q_i \leftarrow partialsum$}
\STATE{{\bf break while}}
\ENDIF
\IF{$partialsum + p_{j} \in (\frac{1}{m}, \frac{3}{2m}]$}
\STATE{$q_i \leftarrow partialsum+p_{j}, \, end[i] \leftarrow j$}
\ELSE \IF{$partialsum + p_{j} > \frac{2}{m}$}
\IF{$partialsum > 0$}
\STATE{$q_i \leftarrow partialsum, \, end[i] \leftarrow j-1, \, i \leftarrow i+1$}
\ENDIF
\STATE{$q_{i}  \leftarrow p_{j}, \, start[i] \leftarrow j, \, end[i] \leftarrow j$}
\ELSE \STATE{// we are left with the case $partialsum+p_j \in (\frac{3}{2m}, \frac{2}{m})$}
\STATE{$q_i \leftarrow partialsum + p_{j}, \, end[i] \leftarrow j$}
\IF{$partialsum > 0$}
\STATE{{\bf Add} index $i$ to the list  {\em Marked-indices}: $mark[i] \leftarrow 1$}
\ENDIF
\ENDIF
\ENDIF
\STATE{$j \leftarrow j+1$} \label{algo2:phase1-end}
\ENDWHILE
\STATE{// At this point $i$ counts the number of components in $\q$}
\STATE{// If $i < m$ we are going to split exactly $m-i$ components starting with the list of {\em Marked-indices}}
\STATE{$k \leftarrow m-i, \, j \leftarrow 1$}   \label{algo2:phase2-start}
\STATE{{\bf Sort} the set ${\cal Q}^*$ of marked components, in non-increasing order}
\STATE{Split $m-i$ largest components in ${\cal Q}^*$. The split is done by creating one component with the largest/last piece and one component with the remaining parts. If  in ${\cal Q}^*$ there are less than $m-i$ components complete with composite components}
\end{algorithmic}
\caption{Improved approximation in $O(n + m\log m)$ time}
\label{algo:Greedy2}
\end{algorithm}

\end{document}